\documentclass{llncs}

\usepackage{latexsym}
\usepackage{bm}
\usepackage{amsfonts}
\usepackage{amsmath}
\usepackage{amssymb}
\usepackage{varwidth}
\usepackage[ruled, vlined, linesnumbered]{algorithm2e}
\usepackage{times}

\usepackage[left]{lineno}
\def\mvar{\mathcal{X}}
\def\mdom{\mathcal{D}}
\def\mcons{\mathcal{C}}

\newcommand{\tuple}[1]{{\bf #1}}
\newcommand{\algo}[1]{\mbox{\texttt{#1}}}

\begin{document}

\mainmatter

\title{The Dichotomy for Conservative Constraint Satisfaction is Polynomially Decidable\thanks{supported by ANR Project ANR-10-BLAN-0210.}}

\author{Cl\'ement Carbonnel}
\institute{CNRS, LAAS, 7 avenue du colonel Roche, F-31400 Toulouse, France\\ University of Toulouse, INP Toulouse, LAAS, F-31400 Toulouse, France
\email{carbonnel@laas.fr}
}

\maketitle

\begin{abstract}
Given a fixed constraint language $\Gamma$, the conservative CSP over $\Gamma$ (denoted by c-CSP($\Gamma$)) is a variant of CSP($\Gamma$) where the domain of each variable can be restricted arbitrarily. In~\cite{DBLP:journals/tocl/Bulatov11} a dichotomy has been proven for conservative CSP: for every fixed language $\Gamma$, c-CSP($\Gamma$) is either in P or NP-complete. However, the characterization of conservatively tractable languages is of algebraic nature and the recognition algorithm provided in~\cite{DBLP:journals/tocl/Bulatov11} is super-exponential in the domain size.
The main contribution of this paper is a polynomial-time algorithm that, given a constraint language $\Gamma$ as input, decides if c-CSP($\Gamma$) is tractable. In addition, if $\Gamma$ is proven tractable the algorithm also outputs its \textit{coloured graph}, which contains valuable information on the structure of $\Gamma$.
\end{abstract}

\vspace{10mm}

\section{Introduction}

The Constraint Satisfaction Problem (CSP) is a powerful framework for solving combinatorial problems, with many applications in artificial intelligence. A CSP instance is a set of variables, a set of values (the \textit{domain}) and a set of constraints, which are relations imposed on a subset of variables. The goal is to assign to each variable a domain value in such a way that all constraints are satisfied. This problem is NP-complete in general.

A very active and fruitful research topic is the non-uniform CSP, in which a set of relations $\Gamma$ is fixed and every constraint must be a relation from $\Gamma$. For instance, if $\Gamma$ contains only binary Boolean relations then CSP($\Gamma$) is equivalent to $2$-SAT and hence polynomially solvable, but if all ternary clauses are allowed the problem becomes NP-complete. The Feder-Vardi Dichotomy Conjecture states that for every finite $\Gamma$, CSP($\Gamma$) is either in P or NP-complete~\cite{Feder98:monotone} (hence missing all the NP-intermediate complexity classes predicted by Ladner's Theorem~\cite{Ladner75:structure}). 

While this conjecture is still open, a major milestone was reached with the characterization of all tractable \textit{conservative} constraint languages, that is, languages that contain every possible unary relation over their domain~\cite{DBLP:journals/tocl/Bulatov11}. Conservativity is a very natural property since it corresponds to the languages that allow arbitrary restrictions of variables domains, a widely used feature in practical constraint solving. It also includes as a particular case the well-studied problem List H-Colouring for a fixed digraph $H$.

Now that the criterion for the tractability of conservative languages has been established, an important question that arises is the complexity of \textit{deciding} if a given conservative language is tractable. An algorithm that decides this criterion efficiently could be used for example as a preprocessing operation in general-purpose constraint solvers, and prompt the use of a dedicated algorithm instead of backtracking search if the instance is over a conservative tractable language.

This meta-problem can be phrased in two slightly different ways. The first would take the whole language $\Gamma$ as input and ask if CSP($\Gamma$) is tractable. However, conservative languages always contain a number of unary relations that is exponential in the domain size, which inflates greatly the input size for the meta-problem without adding any computational difficulty. A more interesting question would take as input a language $\Gamma$ and ask if c-CSP($\Gamma$) is tractable, where c-CSP($\Gamma$) allows all unary relations in addition to $\Gamma$ (this is the \textit{conservative CSP} over $\Gamma$). Designing a polynomial-time algorithm for this meta-problem is more challenging, but it would perform much better as a structural analysis tool for preprocessing CSP instances.

Bulatov's characterization of conservative tractable languages is based on the existence of closure operations (called \textit{polymorphisms}) that satisfy a certain set of identities. While the algebraic nature of this criterion makes the meta-problem delicate to solve, it also shows that the meta-problem is in NP and can be solved in polynomial time if the domain size is fixed. This hypothesis is however very strong because there is only a finite number of constraint languages of fixed arity over a fixed domain. If the domain is not fixed this algorithm becomes super-exponential, and hence is polynomial for neither flavour of the meta-problem.

The contribution of our paper is twofold:
\begin{itemize}
\item[(i)] We present an algorithm that decides the dichotomy for c-CSP in polynomial time. This is the main result of this paper.

\item[(ii)] As a byproduct, we exhibit a general connection between the complexity of the meta-problem and the existence of a \textit{semiuniform algorithm} on classes of conservative languages defined by certain algebraic identities known as \textit{linear strong Mal'tsev conditions}. We obtain as a corollary a broad generalization of the result about conservative Mal'tsev polymorphisms found in~\cite{carbonnel:hal-01230681}.
\end{itemize}

The necessary background for our proofs will be given in Section~\ref{sec:prel}. In Section~\ref{sec:semicons} we will then present the proof of the contribution $(ii)$, and in Section~\ref{sec:main} we will show how this result can be used to derive an algorithm that decides the dichotomy for c-CSP in polynomial time. Finally, we will conclude and discuss open problems in Section~\ref{sec:conc}.

\section{Preliminaries}
\label{sec:prel}

\subsection{Constraint Satisfaction Problems}

An instance of the \textit{constraint satisfaction problem} (CSP) is a triple $(\mvar,\mdom,\mcons)$ where $\mvar$ is a set of variables, $\mdom$ is a finite set of values and $\mcons$ is a set of constraints. A \textit{constraint} $C$ of arity $k$ is a pair $(S_C,R_C)$ where $R_C$ is a $k$-ary relation over $\mdom$ and $S_C \in \mvar^k$ is the \textit{scope} of $C$. The goal is to find an assignment $\phi : \mvar \rightarrow \mdom$ such that for all $C=(S_C,R_C) \in \mcons$, $\phi(S_C) \in R_C$. In this definition, variables do not come with individual domains; any variable-specific domain restriction has to be enforced using a unary constraint.

Given a constraint $C=(S_C,R_C)$ and $X_1 \subseteq \mvar$, we denote by $C[X_1]$ the projection of $C$ onto the variables in $X_1$ (which is the empty constraint if $S$ does not contain any variable in $X_1$). The projection of a CSP instance $I$ onto a subset $X_1 \subseteq \mvar$, denoted by $I_{|X_1}$, is obtained by projecting every constraint onto $X_1$ and then removing all variables that do not belong to $X_1$. A \textit{partial solution} to $I$ is a solution (i.e. a satisfying assignment) to $I_{|X_1}$ for some subset $X_1 \subseteq \mvar$. A CSP instance is \textit{1-minimal} if each variable $x \in \mvar$ has an individual domain $D(x)$ (represented as a unary constraint) and the projection onto $\{x\}$ of every constraint $C \in \mcons$ whose scope contains $x$ is exactly $D(x)$. $1$-minimality can be enforced in polynomial time by gradually removing inconsistent tuples from the constraint relations until a fixed point is reached~\cite{Mackworth1977}.

Throughout the paper we shall use $\mathcal{R}(.)$ and $\mathcal{S}(.)$ as operators that return respectively the relation and the scope of a constraint. A \textit{constraint language} over a set $\mdom$ is a set of relations over $\mdom$, and the constraint language $\mathcal{L}(I)$ of a CSP instance $I=(\mvar,\mdom,\mcons)$ is the set $\{\mathcal{R}(C) \; | \; C \in \mcons \}$. Given a constraint language $\Gamma$ over a set $\mdom$, we denote by $\overline{\Gamma}$ the \textit{conservative extension} of $\Gamma$, that is, the language comprised of $\Gamma$ plus all possible unary relations over $\mdom$. Finally, given a constraint language $\Gamma$ we denote by CSP($\Gamma$) (resp. c-CSP($\Gamma$)) the restriction of CSP to instances $I$ such that $\mathcal{L}(I) \subseteq \Gamma$ (resp. $\mathcal{L}(I) \subseteq \overline{\Gamma}$).

The algorithms presented in this paper will take constraint languages as input, and the complexity analysis depends crucially on how relations are encoded. While practical constraint solvers often represent relations intentionally through \textit{propagators}, we shall always assume that every relation is given as an explicit list of tuples (a very common assumption in theoretical papers).

\subsection{Polymorphisms}

Given a constraint language $\Gamma$, the complexity of CSP($\Gamma$) is usually studied through closure operations called polymorphisms. Given an integer $k$ and a constraint language $\Gamma$ over $\mdom$, a $k$-ary operation $f:\mdom^k \rightarrow \mdom$ is a \textit{polymorphism} of $\Gamma$ if for all $R \in \Gamma$ of arity $r$ and $\tuple{t_1},\ldots,\tuple{t_k} \in R$ we have
$$(f(\tuple{t_1}[1],\ldots,\tuple{t_k}[1]),\ldots,f(\tuple{t_1}[r],\ldots,\tuple{t_k}[r])) \in R$$
A polymorphism $f$ is \textit{idempotent} if $\forall x \in \mdom$, $f(x,\ldots,x) = x$ and \textit{conservative} if $\forall x_1,\ldots,x_k \in \mdom$, $f(x_1,\ldots,x_k) \in \{x_1,\ldots,x_k\}$. It is known that given a constraint language $\Gamma$, the complexity of CSP($\Gamma$) is entirely determined by its polymorphisms~\cite{DBLP:journals/jacm/JeavonsCG97}. On the other hand, the conservative polymorphisms of $\Gamma$ are exactly those that preserve all unary relations, and hence determine the complexity of c-CSP($\Gamma$). A binary polymorphism $f$ is a \textit{semilattice} if $\forall x,y,z \in \mdom$, $f(x,x) = x$, $f(x,y) = f(y,x)$ and $f(f(x,y),z) = f(x,f(y,z))$. A \textit{majority} polymorphism is a ternary polymorphism $f$ such that $\forall x,y \in \mdom$, $f(x,x,y) = f(x,y,x) = f(y,x,x) = x$ and a \textit{minority} polymorphism is a ternary polymorphism $f$ such that $\forall x,y \in \mdom$, $f(x,x,y) = f(x,y,x) = f(y,x,x) = y$.

\subsection{Conservative Constraint Satisfaction}

%In the Boolean case, Schaefer's Theorem states that CSP($\Gamma$) is polynomial-time if $\Gamma$ has either a constant polymorphism, a semilattice polymorphism, a majority polymorphism or a minority polymorphism. Otherwise, CSP($\Gamma$) is NP-complete.
In general, if $\Gamma$ is a conservative language and there exists $\{a,b\} \subseteq \mdom$ such that every polymorphism of $\Gamma$ is a projection when restricted to $\{a,b\}$ then CSP($\{R\}$) is polynomially reducible to CSP($\Gamma$)~\cite{Jeavons98:algebraic}, where
$$R = \left( \begin{array}{ccc}
a & b & b \\
b & a & b \\
b & b & a \end{array} \right)$$
It follows that CSP($\Gamma$) is NP-complete as CSP($\{R\}$) is equivalent to $1$-in-$3$ SAT. The Dichotomy Theorem for conservative CSP states that the converse is true: if for every $B
 = \{a,b\} \subseteq \mdom$ there exists a polymorphism $f$ such that $f_{|B}$ is \textit{not} a projection, then c-CSP($\Gamma$) is polynomial-time. By Post's lattice~\cite{Post41}, the polymorphism $f$ can be chosen such that $f_{|B}$ is either a majority operation, a minority operation or a semilattice.
 
\begin{theorem}[\cite{DBLP:journals/tocl/Bulatov11}]
\label{thm:consdich}
Let $\Gamma$ be a fixed constraint language over a domain $\mdom$. If for every $B = \{a,b\} \subseteq \mdom$ there exists a conservative polymorphism $f$ such that $f_{|B}$ is either a majority operation, a minority operation or a semilattice then c-CSP($\Gamma$) is in P. Otherwise, c-CSP($\Gamma$) is NP-complete.
\end{theorem}

This theorem provides a way to determine the complexity of c-CSP($\Gamma$), since we can enumerate all ternary operations over $\mdom$ and list those that are polymorphisms of $\Gamma$. However, this procedure is super-exponential in time if the domain is part of the input. Our paper presents a more elaborate, polynomial-time algorithm that does not impose any restriction on $\Gamma$.

Three different proofs of Theorem~\ref{thm:consdich} have been published~\cite{DBLP:journals/tocl/Bulatov11}\cite{DBLP:conf/lics/Barto11}\cite{BulatovConsShorter}, and two of them rely heavily on a construction called the \textit{coloured graph} of $\Gamma$ and denoted by G$_{\Gamma}$. The definition of G$_{\Gamma}$ is as follows. The vertex set of G$_{\Gamma}$ is $\mdom$, and there is an edge between any two vertices. Each edge $(a,b)$ is labelled with a colour following these rules:
\begin{itemize}
\item If there exists a polymorphism $f$ such that $f_{|\{a,b\}}$ is a semilattice, then $(a,b)$ is red;
\item If there exists a polymorphism $f$ such that $f_{|\{a,b\}}$ is a majority operation and $(a,b)$ is not red, then $(a,b)$ is yellow;
\item If there exists a polymorphism $f$ such that $f_{|\{a,b\}}$ is a minority operation and $(a,b)$ is neither red nor yellow, then $(a,b)$ is blue.
\end{itemize}

Additionally, red edges are directed: we have $(a \rightarrow b)$ if there exists $f$ such that $f(a,b) = f(b,a) = b$. It is possible to have $(a \leftrightarrow b)$.
By Theorem~\ref{thm:consdich}, G$_{\Gamma}$ is entirely coloured if and only if c-CSP($\Gamma$) is tractable. The next theorem, from~\cite{DBLP:journals/tocl/Bulatov11}, shows that the tractability of c-CSP($\Gamma$) is always witnessed by three specific polymorphisms (instead of $O(d^2)$ in the original formulation).

\begin{theorem}[The Three Operations Theorem~\cite{DBLP:journals/tocl/Bulatov11}]
\label{thm:three}
Let $\Gamma$ be a language such that c-CSP($\Gamma$) is tractable. There exist three conservative polymorphisms $f^*(x,y)$, $g^*(x,y,z)$ and $h^*(x,y,z)$ such that for every two-element set $B \subseteq \mdom$:
\begin{itemize}
\item $f^*_{|B }$ is a semilattice operation if $B$ is red and $f^*(x,y) = x$ otherwise ;
\item $g^*_{|B }$ is a majority operation if $B$ is yellow, $g^*_{|B }(x,y,z) = x$ if $B$ is blue and $g^*_{|B }(x,y,z) = f^*(f^*(x,y),z)$ if $B$ is red ;
\item $h^*_{|B }$ is a minority operation if $B$ is blue, $h^*_{|B }(x,y,z) = x$ if $B$ is yellow, and $h^*_{|B }(x,y,z) = f^*(f^*(x,y),z)$ if $B$ is red.
\end{itemize}
\end{theorem}

The original theorem also proves the existence of other polymorphisms, but we will only use $f^*$, $g^*$ and $h^*$ in our proofs.

\subsection{Meta-problems and identities}
\label{sec:ident}

Given a class T of constraint languages, the \textit{meta-problem} (or \textit{metaquestion}~\cite{chen2016asking}) for T takes as input a constraint language $\Gamma$ and asks if $\Gamma \in T$. In the context of CSP and c-CSP, the class $T$ is often defined as the set of all languages that admit a combination of polymorphisms satisfying a certain set of identities; in this case the meta-problem is a \textit{polymorphism detection problem}. We will be interested in particular sets of identities called \textit{linear strong Mal'tsev conditions}. Given that universal algebra is not the main topic of our paper, we will use a simplified exposition similar to that found in~\cite{chen2016asking}. A \textit{linear identity} is an expression of the form $f(x_1,\ldots,x_k) \approx g(y_1,\ldots,y_c)$ or $f(x_1,\ldots,x_k) \approx y_i$ where $f,g$ are operation symbols and $x_1,\ldots,x_k,y_1,\ldots,y_c$ are variables. It is \textit{satisfied} by two interpretations for $f$ and $g$ on a domain $\mdom$ if the equality holds for any assignment to the variables. A \textit{strong linear Mal'tsev condition} $\mathcal{M}$ is a finite set of linear identities. We say that a set of operations satisfy $\mathcal{M}$ if they satisfy every identity in $\mathcal{M}$. A strong linear Mal'tsev condition is said to be \textit{idempotent} if it entails $f_i(x,\ldots,x) \approx x$ for all operation symbols $f_i$. For a given linear strong Mal'tsev condition, the number of operation symbols and their maximum arity are constant.

\begin{example}
The set of identities
\begin{align*}
f(x,x,y) &\approx x\\
f(x,y,x) &\approx x\\
f(y,x,x) &\approx x
\end{align*}
is the idempotent linear strong Mal'tsev condition that defines majority operations. On the other hand, recall that semilattices are binary operations $f$ satisfying
\begin{align*}
f(x,x) &\approx x\\
f(x,y) &\approx f(y,x)\\
f(x,f(y,z)) &\approx f(f(x,y),z)
\end{align*}
which does not form a linear strong Mal'tsev condition because the identity enforcing the associativity of $f$ is not linear.
\end{example}

By extension, we say that a constraint language satisfies a linear strong Mal'tsev condition $\mathcal{M}$ if it has a collection of polymorphisms that satisfy $\mathcal{M}$. The definability of a class of constraint languages by a linear strong Mal'tsev condition $\mathcal{M}$ is strongly tied up with the meta-problem, because for such classes we can associate any constraint language $\Gamma$ with a polynomial-sized CSP instance whose solutions, if any, are exactly the polymorphisms of $\Gamma$ satisfying $\mathcal{M}$~\cite{chen2016asking}. We will describe the construction below.

Given a constraint language $\Gamma$ and an integer $k$ the \textit{indicator problem} of order $k$ of $\Gamma$, denoted by $\mathcal{IP}^k(\Gamma)$, is a CSP instance with one variable $x_{f(d_1,\ldots,d_k)}$ for every $(d_1,\ldots,d_k) \in \mdom^k$ and one constraint $C^{R^*}_{\tuple{f}(\tuple{t_1},\ldots,\tuple{t_k})}$ for each $R^* \in \Gamma$, $\tuple{t_1},\ldots,\tuple{t_k} \in R^*$. The constraint $C^{R^*}_{\tuple{f}(\tuple{t_1},\ldots,\tuple{t_k})}$ has $R^*$ as relation, and its scope $S$ is such that for all $i \leq |S|$, $S[i] = x_{f(\tuple{t_1}[i],\ldots,\tuple{t_k}[i])}$. Going back to the definition of a polymorphism, it is simple to see that the solutions to $\mathcal{IP}^k(\Gamma)$ are exactly the $k$-ary polymorphisms of $\Gamma$~\cite{DBLP:journals/jacm/JeavonsCG97}.

Now, let $\mathcal{M}$ denote a linear strong Mal'tsev condition with symbols $f_1,\ldots,f_m$ of respective arities $a_1,\ldots,a_m$. We build a CSP instance $\mathcal{P}_{\mathcal{M}}(\Gamma)$ that is the disjoint union of $\mathcal{IP}^{a_1}(\Gamma),\ldots,\mathcal{IP}^{a_m}(\Gamma)$. By construction, each solution $\phi$ to $\mathcal{P}_{\mathcal{M}}(\Gamma)$ is a collection of polymorphisms $(f_1,\ldots,f_m)$ of $\Gamma$. We can force these polymorphisms to satisfy the identities in $\mathcal{M}$ by adding new constraints. If $\mathcal{E}_i \in \mathcal{M}$ is of the form $f_j(x_1,\ldots,x_{a_j}) \approx f_p(y_1,\ldots,y_{a_p})$, we add an equality constraint between the variables $x_{f_j(\phi(x_1),\ldots,\phi(x_{a_j}))}$ and $x_{f_p(\phi(y_1),\ldots,\phi(y_{a_p}))}$ for every possible assignment $\phi$ to $\{x_1,\ldots,x_{a_j},y_1,\ldots,y_{a_p}\}$. Otherwise (i.e. if $\mathcal{E}_i$ is of the form $f_j(x_1,\ldots,x_k) \approx y_i$) we can enforce $\mathcal{E}_i$ by adding unary constraints. Note that the language of $\mathcal{P}_{\mathcal{M}}(\Gamma)$ is $\Gamma$ together with possible equalities and unary relations with a single tuple. This construction will be used frequently throughout the paper.

\subsection{Uniform and semiuniform algorithms}
\label{sec:unif}

Let $\mathcal{M}$ denote a strong linear Mal'tsev condition, and let CSP($\mathcal{M}$) denote the CSP restricted to instances whose language satisfies $\mathcal{M}$.

\begin{definition}
A uniform polynomial-time algorithm for $\mathcal{M}$ is an algorithm that solves CSP($\mathcal{M}$) in polynomial time.
\end{definition}

The term ``uniform" here refers to the fact that the language is not fixed (as in the Feder-Vardi Dichotomy conjecture), but may only range over languages that satisfy $\mathcal{M}$. The existence of a uniform algorithm implies that CSP($\Gamma$) is in P for every $\Gamma$ that satisfies $\mathcal{M}$, but the converse is not guaranteed to be true. For instance, an algorithm for CSP($\mathcal{M}$) that is exponential only in the domain size is polynomial for every fixed $\Gamma$ that satisfies $\mathcal{M}$, but is not uniform. A weaker notion of uniformity called \textit{semiuniformity} has been recently introduced in~\cite{chen2016asking}, and will be central to our paper.

\begin{definition}
A semiuniform polynomial-time algorithm for $\mathcal{M}$ is an algorithm that solves CSP($\mathcal{M}$) in polynomial time provided each instance $I$ is paired with polymorphisms $f_1,\ldots,f_m$ of $\mathcal{L}(I)$ that satisfy $\mathcal{M}$.
\end{definition}

Observe that semiuniform algorithms are tied to the identities in $\mathcal{M}$ rather than the class of languages it defines; even if CSP($\mathcal{M}_1$) and CSP($\mathcal{M}_2$) denote the exact same set of instances, the polymorphisms satisfying $\mathcal{M}_2$ can be more computationally useful than those satisfying $\mathcal{M}_1$.

The following observation has been part of the folklore for some time (see e.g.~\cite{bessiere2013detecting}\cite{barto2014collapse}) and has been recently formalized in~\cite{chen2016asking}.

\begin{proposition}[\cite{chen2016asking}]
\label{prp:chen}
Let $\mathcal{M}$ be an idempotent strong linear Mal'tsev condition. If $\mathcal{M}$ has a uniform algorithm, then the meta-problem for $\mathcal{M}$ is polynomial time.
\end{proposition}

We give here the proof sketch. The idempotency of $\mathcal{M}$ ensures that we have a uniform algorithm for the \textit{search} problem (i.e. decide if the instance is satisfiable and produce a solution if one exists) because idempotent polymorphisms always preserve assignments to variables, which can be seen as unary relations with a single tuple. Given a relational structure $\Gamma$, to check if $\Gamma$ satisfies $\mathcal{M}$ we build the instance $\mathcal{P}_{\mathcal{M}}(\Gamma)$ as in Section~\ref{sec:ident} and invoke the uniform search algorithm. Since the language of $\mathcal{P}_{\mathcal{M}}(\Gamma)$ is $\Gamma$ plus equalities and unary relations with a single tuple, $\mathcal{L}(\mathcal{P}_{\mathcal{M}}(\Gamma))$ satisfies $\mathcal{M}$ if and only if $\Gamma$ does. If $\mathcal{P}_{\mathcal{M}}(\Gamma)$ is satisfiable then $\Gamma$ satisfies $\mathcal{M}$ and the algorithm must produce a solution (which can be easily verified), and whenever the algorithm fails to do so we can safely conclude that $\Gamma$ does not satisfy $\mathcal{M}$.

There is no intuitive way to make this approach work with semiuniform algorithms because they will not run unless given an explicit solution to $\mathcal{P}_{\mathcal{M}}(\Gamma)$ beforehand.

\section{Semiuniformity in Conservative Constraint Languages}
\label{sec:semicons}

As seen in Section~\ref{sec:unif}, in the case of idempotent linear strong Mal'tsev conditions a uniform algorithm implies the tractability of the meta-problem.  We will see that if the problem is to decide if $\overline{\Gamma}$ satisfies $\mathcal{M}$ (i.e. to decide if $\Gamma$ has \textit{conservative} polymorphisms $f_1,\ldots,f_m$ that satisfy $\mathcal{M}$) then semiuniformity is sufficient. This implies that, surprisingly, \textit{uniformity and semiuniformity are equivalent} for classes of conservative languages definable by a strong linear Mal'tsev condition.

The general strategy to solve the meta-problem assuming a semiuniform algorithm is to cast the meta-problem as a CSP and then compute successively partial solutions $\phi_1,\ldots,\phi_\alpha$ of slowly increasing size until a solution to the whole CSP is obtained. The originality of our approach is that $\phi_{i+1}$ is not computed directly from $\phi_i$, but by solving a polynomial number of CSP instances whose languages admit $\phi_i$ as a polymorphism. This algorithm can be seen as a treasure hunt, where each chest contains the key to open the next one.

Let $\mathcal{M}$ be a strong linear Mal'tsev condition with operation symbols $f_1,\ldots,f_m$ of respective arities $a_1,\ldots,a_m$. Let $\Gamma$ be a constraint language over $\mdom$ and $\mathcal{P}_\mathcal{M}(\Gamma)$ be the CSP whose solutions are exactly the polymorphisms of $\Gamma$ satisfying $\mathcal{M}$ (as described in Section~\ref{sec:ident}). Recall that for every symbol $f_i$ in $\mathcal{M}$ and $(d_1,\ldots,d_{a_i}) \in \mdom^{a_i}$ we have a variable $x_{f_i(d_1,\ldots,d_{a_i})}$ that dictates how $f_i$ should map $d_1,\ldots,d_{a_i}$, and for every $R^* \in \Gamma$ and $a_i$ tuples $\tuple{t_1},\ldots,\tuple{t_{a_i}} \in R^*$ we have a constraint $C^{R^*}_{\tuple{f_i}(\tuple{t_1},\ldots,\tuple{t_{a_i}})}$ that forces the tuple $\tuple{f_i}(\tuple{t_1},\ldots,\tuple{t_{a_i}})$ to belong to $R^*$ (where $\tuple{f_i}$ is the operation on tuples obtained by componentwise application of $f_i$). Our goal is to decide if $\overline{\Gamma}$ satisfies $\mathcal{M}$, which requires the polymorphisms of $\Gamma$ satisfying $\mathcal{M}$ to be conservative. The solutions to $\mathcal{P}_\mathcal{M}(\Gamma)$ can easily be guaranteed to be conservative by adding the unary constraint $x_{f_i(d_1,\ldots,d_{a_i})} \in \{ d_1,\ldots,d_{a_i} \}$ on each variable $x_{f_i(d_1,\ldots,d_{a_i})} \in \mathcal{X}$. We will denote this new problem by $\mathcal{P}^c_\mathcal{M}(\Gamma)$, and each solution $\phi$ to $\mathcal{P}_\mathcal{M}^c(\Gamma)$ is a collection $(f_1,\ldots,f_m)$ of conservative polymorphisms of $\Gamma$ satisfying $\mathcal{M}$.

We need one more definition. Given a CSP instance $\mathcal{I}$, a \textit{consistent restriction} of $\mathcal{I}$ is an instance obtained from $\mathcal{I}$ by adding new constraints that are either unary or equalities and then enforcing 1-minimality. We will be interested in the consistent restrictions of $\mathcal{P}_\mathcal{M}^c(\Gamma)$, and we will keep the same notations for constraints that already existed in $\mathcal{P}_\mathcal{M}^c(\Gamma)$. The next lemma is a variation of (\cite{carbonnel:hal-01230681}, Observation 2) adapted to our purpose.

%For the remainder of this section we shall assume that generalized arc-consistency (GAC) has been enforced on $\mathcal{P}_\mathcal{M}^c(\Gamma) = (\mathcal{X},\mathcal{D},\mathcal{C})$. 

\begin{lemma}
\label{lem:easy}
Let $\mathcal{P} = (\mvar,\mdom,\mcons)$ be a consistent restriction of $\mathcal{P}_\mathcal{M}^c(\Gamma)$. Let $f_i$ and $f_j$ be operation symbols in $\mathcal{M}$. If $C^{R^*}_{\tuple{f_i}(\tuple{t_1},\ldots,\tuple{t_{a_i}})} \in \mathcal{C}$ and $\tuple{t'_1},\ldots,\tuple{t'_{a_j}} \in \mathcal{R}(C^{R^*}_{\tuple{f_i}(\tuple{t_1},\ldots,\tuple{t_{a_i}})})$ then 
$$\mathcal{R}(C^{R^*}_{\tuple{f_j}(\tuple{t'_1},\ldots,\tuple{t'_{a_j}})}) \subseteq \mathcal{R}(C^{R^*}_{\tuple{f_i}(\tuple{t_1},\ldots,\tuple{t_{a_i}})})$$
\end{lemma}

\begin{proof}
Let $S = \mathcal{S}(C^{R^*}_{\tuple{f_i}(\tuple{t_1},\ldots,\tuple{t_{a_i}})})$ and $S' = \mathcal{S}(C^{R^*}_{\tuple{f_j}(\tuple{t'_1},\ldots,\tuple{t'_{a_j}})})$. Before 1-minimality was enforced, we had $\mathcal{R}(C^{R^*}_{\tuple{f_i}(\tuple{t_1},\ldots,\tuple{t_{a_i}})}) = \mathcal{R}(C^{R^*}_{\tuple{f_j}(\tuple{t'_1},\ldots,\tuple{t'_{a_j}})}) = R^*$. Thus, after enforcing 1-minimality we have $\mathcal{R}(C^{R^*}_{\tuple{f_i}(\tuple{t_1},\ldots,\tuple{t_{a_i}})}) =  R^* \cap (\pi_{x \in S}D(x))$ and $\mathcal{R}(C^{R^*}_{\tuple{f_j}(\tuple{t'_1},\ldots,\tuple{t'_{a_j}})}) =  R^* \cap (\pi_{x \in S'}D(x))$. However, since $\tuple{t'_1},\ldots,\tuple{t'_{a_j}} \in \mathcal{R}(C^{R^*}_{\tuple{f_i}(\tuple{t_1},\ldots,\tuple{t_{a_i}})})$, the conservativity constraints ensure that for each $k$, $$D(S'[k]) = D(x_{f_j(\tuple{t'_1}[k],\ldots,\tuple{t'_{a_j}}[k])}) \subseteq \{ \tuple{t'_1}[k],\ldots,\tuple{t'_{a_j}}[k] \} \subseteq     D(S[k])$$ 
Therefore, $\mathcal{R}(C^{R^*}_{\tuple{f_j}(\tuple{t'_1},\ldots,\tuple{t'_{a_j}})}) \subseteq \mathcal{R}(C^{R^*}_{\tuple{f_i}(\tuple{t_1},\ldots,\tuple{t_{a_i}})})$.
\end{proof}

Given two sets of variables $X_1,X_2 \subseteq \mathcal{X}$, we write $X_1 \vartriangleleft X_2$ if for each symbol $f_i$ in $\mathcal{M}$, $\forall x \in X_2$ and $\tuple{t} \in D(x)^{a_i}$ we have $x_{f_i(\tuple{t})} \in X_1$. If $X_1 \vartriangleleft X_1$, we say that $X_1$ is \textit{closed}.

\begin{proposition}
\label{prp:conspoly}
Let $\mathcal{P} = (\mvar,\mdom,\mcons)$ be a consistent restriction of $\mathcal{P}_\mathcal{M}^c(\Gamma)$. If $X_1$ and $X_2$ are subsets of variables such that $X_1 \vartriangleleft X_2$, then every solution to $\mathcal{P}_{|X_1}$ is a collection of polymorphisms of $\mathcal{L}(\mathcal{P}_{|X_2})$.
\end{proposition}

\begin{proof}
Let $f_i,f_j \in \{f_1,\ldots,f_m\}$ be operation symbols in $\mathcal{M}$. Let $R^* \in \Gamma$, $\tuple{t_1},\ldots,\tuple{t_{a_i}} \in R^*$, $C_2 = (S_2,R_2) \in \mathcal{P}_{|X_2}$ be the projection of $C^{R^*}_{\tuple{f_i}(\tuple{t_1},\ldots,\tuple{t_{a_i}})}$ onto $X_2$, and $\tuple{t^2_1},\ldots,\tuple{t^2_{a_j}} \in R_2$. By the nature of projections, there must exist $\tuple{t'_1},\ldots,\tuple{t'_{a_j}} \in \mathcal{R}(C^{R^*}_{\tuple{f_i}(\tuple{t_1},\ldots,\tuple{t_{a_i}})})$ such that $\tuple{t^2_1},\ldots,\tuple{t^2_{a_j}}$ is the projection of  $\tuple{t'_1},\ldots,\tuple{t'_{a_j}}$ onto $X_2$. Then, by Lemma~\ref{lem:easy} we have 
$$\mathcal{R}(C^{R^*}_{\tuple{f_j}(\tuple{t'_1},\ldots,\tuple{t'_{a_j}})}) \subseteq \mathcal{R}(C^{R^*}_{\tuple{f_i}(\tuple{t_1},\ldots,\tuple{t_{a_i}})})$$
and in particular $\mathcal{R}(C^{R^*}_{\tuple{f_j}(\tuple{t'_1},\ldots,\tuple{t'_{a_j}})}[X_2]) \subseteq \mathcal{R}(C^{R^*}_{\tuple{f_i}(\tuple{t_1},\ldots,\tuple{t_{a_i}})}[X_2]) = R_2$. Now, note that because $X_1 \vartriangleleft X_2$ and $\mathcal{P}$ is 1-minimal, every variable $x_{f_j(\tuple{t'_1}[k],\ldots,\tuple{t'_{a_j}}[k])}$ in the scope of $C^{R^*}_{\tuple{f_j}(\tuple{t'_1},\ldots,\tuple{t'_{a_j}})}[X_2]$ also belongs to $X_1$. We denote this constraint by $C_1$.

Let us summarize what we have: for every symbol $f_j$, every relation $R_2 \in \mathcal{L}(\mathcal{P}_{|X_2})$ other than equalities and unary relations (which are preserved by all conservative polymorphisms) and $\tuple{t^2_1},\ldots,\tuple{t^2_{a_j}} \in R_2$, there is a constraint $C_1 = (S_1,R_1) \in \mathcal{P}_{|X_1}$ such that $|S_1| = |S_2|$, $R_1 \subseteq R_2$ and for every $k$ we have $S_1[k] = x_{f_j(\tuple{t^2_1}[k],\ldots,\tuple{t^2_{a_j}}[k])}$. It follows that for every solution $(f_1,\ldots,f_m)$ to $\mathcal{P}(\Gamma)_{|X_1}$, $f_j$ is also a solution to the indicator problem of order $a_j$ of $\mathcal{L}(\mathcal{P}(\Gamma)_{|X_2})$ and is therefore a polymorphism of $\mathcal{L}(\mathcal{P}(\Gamma)_{|X_2})$.
\end{proof}

Closed sets of variables allow us to turn partial solutions into true polymorphisms of a specific constraint language, hence enabling us to make (limited) use of semiuniform algorithms. A variable of $\mathcal{P}^c_{\mathcal{M}}(\Gamma)$ is a \textit{singleton} if it is of the form $x_{f_i(v,\ldots,v)}$ for some $v \in \mdom$. The sets of variables corresponding to singletons and $\mathcal{X}$ constitute two closed sets; the next Lemma shows that many intermediate, regurlarly-spaced closed sets exist in $\mathcal{P}^c_{\mathcal{M}}(\Gamma)$ between these two extremes.

\begin{lemma}
\label{lem:subsets}
Let $\mathcal{P}^c_{\mathcal{M}}(\Gamma) = (\mvar,\mdom,\mcons)$ after applying 1-minimality. There exist $X_0 \subseteq \ldots \subseteq X_\alpha = \mathcal{X}$ such that $X_0$ is the set of all singleton variables, each $X_i$ is closed and $|X_{i+1} - X_i| \leq ma^a$, where $a$ and $m$ denote respectively the maximum arity and number of operation symbols in $\mathcal{M}$.
\end{lemma}

\begin{proof}
Let $(D_1,\ldots,D_\alpha)$ denote an arbitrary ordering of the subsets of $\mdom$ of size $a$. We define
$$X_0 = \{ x_{f_j(v_i,\ldots,v_i)} \, | \, f_j \in \mathcal{M}, v_i \in \mdom \}$$
and for all $i \in [1..\alpha]$
$$X_i = X_{i-1} \cup \{ x_{f_j(\tuple{t})} \, | \, f_j \in \mathcal{M}, \tuple{t} \in (D_{i})^{a_j} \}$$

It is clear that $X_0$ is the set of all singleton variables and for all $i$, $|X_{i+1} - X_i| \leq m|(D_{i})^{a}| = ma^a$. It remains to show that each set is closed. Let $k \geq 1$ and suppose that $X_{k-1}$ is closed. By induction hypothesis, we only need to verify that $X_{k} \vartriangleleft X_{k} \backslash X_{k-1}$. Let $x_{f_j(v_1,\ldots,v_{a_j})}$ be a variable in $X_{k} \backslash X_{k-1}$. Because $\mathcal{P}^c_{\mathcal{M}}(\Gamma)$ is 1-minimal, we have $D(x_{f_j(v_1,\ldots,v_{a_j})}) \subseteq \{v_1,\ldots,v_{a_j}\} \subseteq D_k$. By construction $X_k$ contains all variables of the form $x_{f_c(\tuple{t})}$ where $\tuple{t} \in (D_{k})^{a_c}$ and because $D(x_{f_j(v_1,\ldots,v_{a_j})}) \subseteq \{v_1,\ldots,v_{a_j}\} \subseteq D_k$ it contains in particular all variables $x_{f_c(\tuple{t})}$ such that $t \subseteq D(x_{f_j(v_1,\ldots,v_{a_j})})$. This implies that $X_{k} \vartriangleleft X_{k} \backslash X_{k-1}$ and concludes the proof.
\end{proof}

We now have every necessary tool at our disposal to start solving $\mathcal{P}^c_{\mathcal{M}}(\Gamma)$. It is straightforward to see that if a subset of variables $X'$ is closed in $\mathcal{P}^c_{\mathcal{M}}(\Gamma)$, then it is closed in every consistent restriction as well.

\begin{proposition}
\label{prop:next}
If a solution to $\mathcal{P}^c_{\mathcal{M}}(\Gamma)_{|X_i}$ is known, then a solution to $\mathcal{P}^c_{\mathcal{M}}(\Gamma)_{|X_{i+1}}$ can be found in polynomial time.
\end{proposition}

\begin{proof}
Let $(f_1^i,\ldots,f_m^i)$ be a solution to $\mathcal{P}^c_{\mathcal{M}}(\Gamma)_{|X_i}$. We assume that 1-minimality has been enforced on $\mathcal{P}^c_{\mathcal{M}}(\Gamma)$. This ensures, in particular, that the domain of each $x_{f_j(\tuple{t})} \in X_{i+1} \backslash X_{i}$ contains at most $a$ elements. It follows that $X_{i+1} \backslash X_{i}$ has at most $s = a^{ma^a}$ possible assignments $\phi_1,\ldots,\phi_{s}$. For every $j \in [1..s]$, we create a CSP instance $\mathcal{P}_j$ that is a copy of 
$\mathcal{P}^c_{\mathcal{M}}(\Gamma)$ but also includes the constraints corresponding to the assignment $X_{i+1} \backslash X_{i} \gets \phi_j(X_{i+1} \backslash X_{i})$. We enforce 1-minimality on every instance $\mathcal{P}_j$.

Now, observe that each $\mathcal{P}_j$ is a consistent restriction of $\mathcal{P}^c_{\mathcal{M}}(\Gamma)$, so $X_i$ is still closed in $\mathcal{P}_j$. Moreover, every variable $x \in X_{i+1} \backslash X_i$ has domain size $1$ in $\mathcal{P}_j$; since $X_i$ contains all singleton variables, if follows that in $\mathcal{P}_j$ we have $X_i \vartriangleleft X_{i+1}$. 

%By Proposition~\ref{prp:conspoly}, the solution $(f_1^i,\ldots,f_m^i)$ is a collection of polymorphisms of $\mathcal{L}({\mathcal{P}_j}_{|X_{i+1}})$, and we can use the semiuniform algorithm to decide if ${\mathcal{P}_j}_{|X_{i+1}}$ has a solution. A solution to $\mathcal{P}^c_{\mathcal{M}}(\Gamma)_{|X_{i+1}}$ exists if and only if at least one assignment $\phi_j$ ($j \in \{1,\ldots,s\}$) yields a solution for ${\mathcal{P}_j}_{|X_{i+1}}$.

By Proposition~\ref{prp:conspoly}, $(f_1^i,\ldots,f_m^i)$ is a collection of polymorphisms of $\mathcal{L}({\mathcal{P}_j}_{|X_{i+1}})$. We can then use the semiuniform algorithm to find in polynomial time a solution to ${\mathcal{P}_j}_{|X_{i+1}}$ if one exists by backtracking search (every $f_z^i$ is idempotent, so we can invoke the semiuniform algorithm at each node to ensure that the algorithm cannot backtrack more than one level). A solution to $\mathcal{P}^c_{\mathcal{M}}(\Gamma)_{|X_{i+1}}$ exists if and only if ${\mathcal{P}_j}_{|X_{i+1}}$ has a solution for some $j \in \{1,\ldots,s\}$.
\end{proof}

The above proof balances on the fact that every complete instantiation of the variables in $X_{i+1} \backslash X_{i}$ (followed by 1-minimality) yields a residual instance over a language that admits $(f_1^i,\ldots,f_m^i)$ as polymorphisms. In other terms, $\mathcal{P}^c_{\mathcal{M}}(\Gamma)_{|X_{i+1}}$ has a \textit{backdoor}~\cite{Williams03:backdoors} of constant size to $(f_1^i,\ldots,f_m^i)$.

\begin{theorem}
\label{thm:semiunif}
Let $\mathcal{M}$ be a linear strong Mal'tsev condition that admits a semiuniform algorithm. There exists a polynomial-time algorithm that, given as input a constraint language $\Gamma$, decides if $\overline{\Gamma}$ satisfies $\mathcal{M}$ and produces conservative polymorphisms of $\Gamma$ satisfying $\mathcal{M}$ if any exist.
\end{theorem}

\begin{proof}
The algorithm starts by building $\mathcal{P}^c_{\mathcal{M}}(\Gamma)$ and computes the sets $X_0,\ldots,X_\alpha$ as in Lemma~\ref{lem:subsets}. We have a solution to $\mathcal{P}^c_{\mathcal{M}}(\Gamma)_{|X_0}$ for free because of the conservativity constraints, and we can compute a solution to $\mathcal{P}^c_{\mathcal{M}}(\Gamma)$ by invoking repeatedly (at most $\alpha \leq |\mathcal{X}| \leq md^a$ times) Proposition~\ref{prop:next}.
\end{proof}

\begin{corollary}
If $\mathcal{M}$ is a linear strong Mal'tsev condition that has a semiuniform algorithm for conservative languages, then $\mathcal{M}$ has also a \textit{uniform} algorithm for conservative languages.
\end{corollary}

\begin{proof}
The uniform algorithm simply invokes our algorithm to produce the conservative polymorphisms satisfying $\mathcal{M}$, and then provides these polymorphisms to the semiuniform algorithm to solve the CSP instance.
\end{proof}

An immediate application of Theorem~\ref{thm:semiunif} concerns the detection of conservative $k$-edge polymorphisms for a fixed $k$. A $k$-edge operation on a set $\mdom$ is a $(k+1)$-ary operation $e$ satisfying
\begin{align*}
&e(x, x, y, y, y, \ldots, y, y) \approx y\\
&e(x, y, x, y, y, \ldots, y, y) \approx y\\
&e(x, y, y, x, y, \ldots, y, y) \approx y\\
&e(x, y, y, y, x, \ldots, y, y) \approx y\\
&\hspace{20mm} \ldots \\
&e(x, y, y, y, y, \ldots, x, y) \approx y\\
&e(x, y, y, y, y, \ldots, y, x) \approx y\\
\end{align*}
These identities form a linear strong Mal'tsev condition. The algorithm given in~\cite{DBLP:conf/lics/IdziakMMVW07} is semiuniform, but in addition to $e$ it must have access to three other polymorphisms $p,d,s$ derived from $e$ and satisfying
\begin{align*}
p(x,y,y) &\approx x\\
p(x,x,y) &\approx d(x,y)\\
d(x,d(x,y)) &\approx d(x,y)\\
s(x, y, y, y, \ldots, y, y) &\approx d(y,x)\\
s(y, x, y, y, \ldots, y, y) &\approx y\\
s(y, y, x, y, \ldots, y, y) &\approx y\\
\ldots\\
s(y, y, y, y, \ldots, y, x) &\approx y\\
\end{align*}
The authors provide a method to obtain these three polymorphisms from $e$ that requires a possibly exponential number of compositions. However, conservative algebras are much simpler and we can observe that
\begin{align*}
s(x_1,x_2,\ldots,x_k) &= e(x_2,x_1,x_2,x_3,\ldots,x_k)\\
d(x,y) &= e(x,y,x,\ldots,x)\\
p(x,y,z) &= e(y,d(y,z),x,\ldots,x)
\end{align*}
satisfy the required identities and are easy to compute. It follows that in the conservative case their algorithm is semiuniform even if only a $k$-edge polymorphism $e$ is given.

\begin{corollary}
\label{cor:edge}
For every fixed $k$, the class of constraint languages admitting a conservative $k$-edge polymorphism is uniformly tractable and has a polynomially decidable meta-problem.
\end{corollary}

Since conservative $2$-edge polymorphisms are Mal'tsev polymorphisms, this corollary is a broad generalization of the result obtained in~\cite{carbonnel:hal-01230681} concerning conservative Mal'tsev polymorphisms.

\section{Deciding the Dichotomy}
\label{sec:main}

While the criterion for the conservative dichotomy theorem can be stated as a linear strong Mal'tsev condition~\cite{siggers2010strong}, none of the algorithms found in the literature are semiuniform. Still, Theorem~\ref{thm:semiunif} gives a uniform algorithm for constraint languages $\Gamma$ whose coloured graph contains only yellow and blue edges: if $g^*(x,y,z)$ and $h^*(x,y,z)$ are the polymorphisms predicted by the Three Operations Theorem, then $m^*(x,y,z) = h^*(g^*(x,y,z),g^*(y,z,x),g^*(z,x,y))$ is a generalized majority-minority polymorphism of $\Gamma$ (see~\cite{DBLP:journals/lmcs/Dalmau06} for a formal definition), which implies that $\Gamma$ has a $3$-edge polymorphism~\cite{berman2010varieties}.

Our algorithm will reduce the meta-problem to a polynomial number of CSP instances over languages with conservative $3$-edge polymorphisms using a refined version of the treasure hunt algorithm and a simple reduction rule. This reduction rule is specific to indicator problems and allows us to avoid the elaborate machinery presented in~\cite{BulatovConsShorter} to eliminate red edges in CSP instances over a tractable conservative language.

We start by the reduction rule. Recall that the Three Operations Theorem predicts that if $\Gamma$ is tractable then it has a conservative polymorphism $f^*$ such that for every 2-element set $B$, $f^*_{|B}$ is a semilattice if $B$ is red and $f^*_{|B}(x,y) = x$ otherwise.

\begin{proposition}
\label{prp:yelblue}
If $f^*$ is known, then for every non-red $2$-element subset $B$ of $\mdom$ it can be decided in polynomial time if there exists a conservative polymorphism $p$ such that $p_{|B}$ is a majority (resp. minority) operation.
\end{proposition}

\begin{proof}
We are looking for a ternary polymorphism $p$, so we start by building the instance $\mathcal{IP}^{3c}(\Gamma)$, which is the indicator problem of order $3$ of $\Gamma$ with conservativity constraints. For $i \in \{1,2,3\}$, let $\pi_i$ be the solution to $\mathcal{IP}^{3c}(\Gamma)$ given by $\pi_i(x_{v_1,v_2,v_3}) = v_i$ for all $v_1,v_2,v_3 \in \mdom$. These solutions correspond to the three ternary polymorphisms of $\Gamma$ that project onto their $i$th argument. We enforce 1-minimality and apply the algorithm \algo{Reduce}.

\begin{algorithm}[h!]
\caption{\algo{Reduce}}
\label{alg:qinter2}
\BlankLine
$s_1 \gets \pi_1$ \;
$s_2 \gets \pi_2$ \;
$s_3 \gets \pi_3$ \;

\While {There exist $i,j$ and $x \in \mvar$ such that $\{s_i(x),s_j(x)\}$ is red and $f^*(s_i(x),s_j(x)) = s_j(x)$} {
	$s_1 \gets f^*(s_1,s_j)$ \;
	$s_2 \gets f^*(s_2,s_j)$ \;
	$s_3 \gets f^*(s_3,s_j)$ \;
	\For {all $x \in \mvar$ and $v \in D(x)$ s.t. $\forall k$, $s_k(x) \neq v$} {
		$D(x) \gets D(x) \backslash v$ \;
	}
}
\end{algorithm}

We denote by $\mathcal{IP}^{3c}_R(\Gamma)$ the resulting CSP instance. An important invariant of this algorithm is that at the end of every iteration of the loop in \algo{Reduce}, for every $x \in \mathcal{X}$ and $v \in D(x)$ there exists $s \in \{s_1,s_2,s_3\}$ such that $s(x) = v$. This is straightforward, since we only remove $v$ from $D(x)$ if none of $s_1(x),s_2(x),s_3(x)$ takes value $v$. It then follows from the loop condition that at the end of \algo{Reduce}, no $x \in \mathcal{X}$ may have a domain that contains a red pair of elements.

We now show that if $\mathcal{IP}^{3c}(\Gamma)$ has a solution that is majority (resp. minority) on a non-red pair of values $B$, then so does $\mathcal{IP}_R^{3c}(\Gamma)$. We proceed by induction. Suppose that at iteration $i$ of the loop of \algo{Reduce}, a solution $p_i$ that is majority (resp. minority) on $B$ exists. Let $D_i(x)$ denote the domain of a variable $x$ at step $i$. We set $p_{i+1} = f^*(p_i,s_j)$. Because $f$ always projects onto its first argument on non-red pairs, a value $v$ can only be removed from $D_i(x)$ at iteration $i+1$ if $\{v,s_j(x)\}$ is red and $f(v,s_j(x)) = s_j(x)$. Therefore, if $p_i(x)$ is removed at iteration $i$ then $p_{i+1}(x) = f^*(p_i(x),s_j(x)) = s_j(x)$, and otherwise $p_{i+1}(x) \in \{p_i(x),s_j(x)\} \subseteq D_{i+1}(x)$; in any case $p_{i+1}(x) \in D_{i+1}(x)$. Furthermore, since $B$ is not red, $p_{i+1}(x_{f(v_1,v_2,v_3)}) = p_{i}(x_{f(v_1,v_2,v_3)})$ for all $\{v_1,v_2,v_3\} \subseteq B$ and we can conclude that $p_{i+1}$ is still majority (resp. minority) on $B$.

Now, we enforce 1-minimality again. We can ensure that every solution is a majority (resp. minority) polymorphism when restricted to $B$ by assigning the $6$ variables concerned by the majority (resp. minority) identity. Since the remaining instance $I$ is red-free in $G_{\Gamma}$, either c-CSP($\Gamma$) is intractable or $\mathcal{L}(I)$ admits a $3$-edge polymorphism. We test for the existence of a $3$-edge polymorphism using Theorem~\ref{thm:semiunif}. If one exists we use the uniform algorithm given by Corollary~\ref{cor:edge} to decide if a solution exists and otherwise we can conclude that c-CSP($\Gamma$) is intractable.
\end{proof}

With this result in mind, the last challenge is to design a polynomial-time algorithm that finds a binary polymorphism $f^*$ that is commutative on as many 2-element subsets as possible, and projects onto its first argument otherwise. We call such polymorphisms \textit{maximally commutative}. This can be achieved using a variant of the algorithm presented in Section~\ref{sec:semicons} and the following Lemma.

\begin{lemma}
\label{lem:bin}
Let $\mathcal{P}=(\mvar,\mdom,\mcons)$ denote an 1-minimal instance such that $\forall x \in \mvar$, $|D(x)| \leq 2$. Suppose that we have a conservative binary polymorphism $f$ of $\mathcal{L}(\mathcal{P})$ and a partition $(V_1,V_2)$ of the variables such that $f(a,b)=f(b,a)=f(D(x))$ whenever $x \in V_1$, and $f$ projects onto its first argument otherwise. Then, every variable $x \in V_1$ can be assigned to $f(D(x))$ without altering the satisfiability of $\mathcal{P}$.
\end{lemma}

\begin{proof}
Let $C=(S,R) \in \mathcal{C}$. Let $S_1 = S \cap V_1$, $S_2 = S \cap V_2$ and $\tuple{t} \in R$. We assume without loss of generality that no variable in $S$ is ground (i.e. has a singleton domain). If $x \in S$, let $\overline{\tuple{t}[x]} = D(x) \backslash \tuple{t}[x]$. Because $\mathcal{P}$ is 1-minimal, for every $x \in S_1$ there exists $\tuple{t_x} \in R$ such that $\tuple{t_x}[x] = \overline{\tuple{t}[x]}$. Let $x_1,\ldots,x_s$ denote an arbitrary ordering of $S_1$. Then, let $\tuple{t^{(0)}} = \tuple{t}$ and for $i \in \{1,\ldots,s\}$,
$$\tuple{t^{(i)}} = \tuple{f(t^{(i-1)},t_{x_i})}$$
It is immediate to see that if $x \in S_2$, then $\tuple{t^{(s)}}[x] = \tuple{t}[x]$ since $f$ will project onto its first argument at each interation. On the other hand, if $x_k \in S_1$ and there exists $j$ such that $\tuple{t^{(j)}}[x_k] = f(D(x_k))$ then $\tuple{t^{(i)}}[x_k] = f(D(x_k))$ for all $i \geq j$. This is guaranteed to happen for $j \leq k$, as either 
\begin{itemize}
\item $\tuple{t}[x_k] = f(D(x_k))$, in which case it is true for $j=0$, or
\item $\tuple{t^{(k-1)}}[x_k] = f(D(x_k))$, in which case it is true for $j = k-1$, or
\item $\tuple{t^{(k-1)}}[x_k] = \tuple{t}[x_k] \neq f(D(x_k))$, in which case $\tuple{t^{(k)}}[x_k] = f(\tuple{t^{(k-1)}}[x_k],\tuple{t_{x_k}}[x_k]) = f(\tuple{t}[x_k],\overline{\tuple{t}[x_k]}) = f(D(x_k))$ and thus it is true for $j=k$.
\end{itemize}
It follows that $\tuple{t^{(s)}}$ is a tuple or $R$ that coincides with $\tuple{t}$ on $S_2$, and $\tuple{t^{(s)}}[x] = D(f(x))$ whenever $x \in S_1$. Therefore, assigning each $x \in S_1$ to $D(f(x))$ is always compatible with any assignment to $S_2$; since this is true for each constraint, it is true for $\mathcal{P}$ as well.
\end{proof}

We denote by $\mathcal{IP}^{2c}(\Gamma)$ the CSP instance obtained from $\mathcal{IP}^2(\Gamma)$ by adding the unary constraints enforcing conservativity. We can interpret $\mathcal{IP}^{2c}(\Gamma)$ as the meta-problem associated with an unconstrained conservative binary operation symbol $f$ and reuse the definitions and lemmas about closed sets of variables seen in the last section. In the hierarchy of closed sets given by Lemma~\ref{lem:subsets} applied to $\mathcal{IP}^{2c}(\Gamma)$, $X_{i+1}$ contains the variables of $X_i$ plus two variables $x_{f(a,b)},x_{f(b,a)}$ for some $B_{i+1} = \{a,b\} \subseteq \mathcal{D}$.

\begin{proposition}
\label{prp:maxcom}
Suppose that we know a solution $f_i$ to $\mathcal{IP}^{2c}(\Gamma)_{|X_i}$ that is maximally commutative if c-CSP($\Gamma$) is tractable. A solution $f_{i+1}$ to $\mathcal{IP}^{2c}(\Gamma)_{|X_{i+1}}$ with the same properties can be found in polynomial time.
\end{proposition}

\begin{proof}
The strategy is similar to the proof of Proposition~\ref{prop:next}. The two differences are that we do not have a semiuniform algorithm in general, which can be handled by Lemma~\ref{lem:bin}, and the fact that we are not interested in \textit{any} solution but in one that is maximally commutative.

Observe that if $\mathcal{IP}^{2c}(\Gamma)_{|X_{i+1}}$ is 1-minimal, then its language is conservatively tractable and the order-2 conservative indicator problem of $\mathcal{L}(\mathcal{IP}^{2c}(\Gamma)_{|X_{i+1}})$ is $\mathcal{IP}^{2c}(\Gamma)_{|X_{i+1}}$ itself plus unconstrained variables (because $X_{i+1}$ is closed). Therefore, by the Three Operations Theorem, a maximally commutative solution to $\mathcal{IP}^{2c}(\Gamma)_{|X_{i+1}}$ is commutative on some $\{u,v\}$ if and only if there is a solution to $\mathcal{IP}^{2c}(\Gamma)_{|X_{i+1}}$ that is also commutative on $\{u,v\}$. It follows from this same argument applied to $X_i$ instead of $X_{i+1}$ that if $f_i$ is \textit{not} commutative on some $(u,v) \in \mdom^2$ then either c-CSP($\Gamma$) is NP-complete or $\Gamma$ has a ternary conservative polymorphism $p_{u,v}$ that is either a majority or a minority operation on $\{u,v\}$.

Let $X_{i+1} = X_i \cup \{x_{f(a,b)},x_{f(b,a)}\}$. We have only three assignments to examine for $(x_{f(a,b)},x_{f(b,a)})$: $(a,a)$, $(b,b)$ and $(a,b)$. The fourth assignment $(b,a)$ is the projection onto the second argument, which does not need to be tried since we are only interested in the maximally commutative solutions to $\mathcal{IP}^{2c}(\Gamma)_{|X_{i+1}}$.
For each of these assignments, we build the CSP instances $\mathcal{P}^1,\mathcal{P}^2,\mathcal{P}^3$ by adding the constraints corresponding to the possible assignments to $(x_{f(a,b)},x_{f(b,a)})$ to $\mathcal{IP}^{2c}(\Gamma)$ and enforcing 1-minimality.

%Since the variables in $X_{i+1} \backslash X_i$ are ground in $\mathcal{P}$, $X_i$ is closed and $X_i$ contains all singleton variables, we have $X_{i+1} \vartriangleleft X_i$ in $\mathcal{P}$. By Proposition~\ref{prp:conspoly}, $f_i$ is a polymorphism of $\mathcal{L}(\mathcal{P}_{|X_{i+1}})$. 

For every $j \in \{1,2,3\}$ and every pair $\{u,v\}$ of elements in the domain of $\mathcal{P}^j_{|X_{i+1}}$ we create an instance $\mathcal{P}^j_{uv}$ by adding the constraint $x_{f(u,v)} = x_{f(v,u)}$ to $\mathcal{P}^j$ and enforcing 1-minimality. Since the variables in $X_{i+1} \backslash X_i$ are ground in $\mathcal{P}^j_{uv}$, $X_i$ is closed and $X_i$ contains all singleton variables, we have $X_{i+1} \vartriangleleft X_i$ in $\mathcal{P}^j_{uv}$. By Proposition~\ref{prp:conspoly}, $f_i$ is a polymorphism of $\mathcal{L}({\mathcal{P}^j_{uv}}_{|X_{i+1}})$.  Now, if a variable $x$ in ${\mathcal{P}^j_{uv}}_{|X_{i+1}}$ has domain size $2$ and $f_i$ is commutative on $D(x)$, by Lemma~\ref{lem:bin} we can assign $x$ to $f_i(D(x))$ without losing the satisfiability of the instance. Once this is done, we can enforce 1-minimality again; the polymorphisms $p_{u',v'}$ guarantee that if c-CSP($\Gamma$) is tractable, the remaining instance has a conservative generalized majority-minority polymorphism and hence a conservative $3$-edge polymorphism. Using Corollary~\ref{cor:edge}, we can decide if the language of ${\mathcal{P}^j_{uv}}_{|X_{i+1}}$ has a conservative $3$-edge polymorphism. If it does not then c-CSP($\Gamma$) is NP-complete, and otherwise we can decide if a solution exists in polynomial time.

%Now, observe that if $\mathcal{IP}^{2c}(\Gamma)_{|X_{i+1}}$ is generalized arc-consistent, then its language is tractable and the order-2 indicator problem of $\mathcal{L}(\mathcal{IP}^{2c}(\Gamma)_{|X_{i+1}})$ is $\mathcal{IP}^{2c}(\Gamma)_{|X_{i+1}}$ itself. Therefore, by the Three Operations Theorem, there exists a solution to $\mathcal{IP}^{2c}(\Gamma)_{|X_{i+1}}$ that is commutative on some $\{u,v\}$, if and only if there is a maximally commutative solution to some $\mathcal{P}_{|X_{i+1}}$ that is also commutative on $\{u,v\}$. 

At this point, for every pair $(u,v)$ of elements in the domain of some variable in $\mathcal{IP}^{2c}(\Gamma)_{|X_{i+1}}$ we know if a solution to $\mathcal{IP}^{2c}(\Gamma)_{|X_{i+1}}$ that is commutative on $(u,v)$ exists, except if $(u,v) = (a,b)$. This problem can be fixed by checking if any of $\mathcal{P}^k_{|X_{i+1}}$ or $\mathcal{P}^n_{|X_{i+1}}$ has a solution, where $\mathcal{P}^k$ and $\mathcal{P}^n$ are the subproblems corresponding to the assignments $(x_{f(a,b)},x_{f(b,a)}) \gets (a,a)$ and $(x_{f(a,b)},x_{f(b,a)}) \gets (b,b)$.

We then add the equality constraint $x_{f(u,v)} = x_{f(v,u)}$ to $\mathcal{IP}^{2c}(\Gamma)_{|X_{i+1}}$ for \textit{every} pair $(u,v)$ (including $(a,b)$ if applicable) such that a solution to $\mathcal{IP}^{2c}(\Gamma)_{|X_{i+1}}$ that is commutative on $(u,v)$ exists. On all other pairs, we know that $f_{i+1}$ must project on the first argument, so we can ground the corresponding variables. If c-CSP($\Gamma$) is tractable, then this new CSP instance $\mathcal{P}$ has a solution and it must be maximally commutative. We can solve $\mathcal{P}$ by branching on the possible assignments to $(x_{f(a,b)},x_{f(b,a)})$ and the usual arguments using $f_i$, Proposition~\ref{prp:conspoly} and Lemma~\ref{lem:bin}.
\end{proof}

\begin{theorem}
There exists a polynomial-time algorithm \algo{A} that, given in input a constraint language $\Gamma$, decides if c-CSP($\Gamma$) is in P or NP-complete. If c-CSP($\Gamma$) is in P, then \algo{A} also returns the coloured graph of $\Gamma$.
\end{theorem}

\begin{proof}
We use Proposition~\ref{prp:maxcom} to find in polynomial time a conservative polymorphism $f^*$ of $\Gamma$ that is maximally commutative if c-CSP($\Gamma$) is tractable. If the algorithm fails, then we know that c-CSP($\Gamma$) is not tractable and the algorithm stops. Otherwise, we label every pair $\{a,b\}$ of domain elements with the colour red if $f^*$ is commutative on $\{a,b\}$, and otherwise we use Proposition~\ref{prp:yelblue} to check if there is a conservative ternary polymorphism that is either majority or minority on $\{a,b\}$. If a majority polymorphism is found then we label $\{a,b\}$ with yellow, else if a minority polymorphism is found then $\{a,b\}$ is blue, and otherwise we know that c-CSP($\Gamma$) is NP-complete. The orientation of the red edges can be easily computed from $\mathcal{IP}^{2c}(\Gamma)$ using Lemma~\ref{lem:bin} and $f^*$.
\end{proof}

\section{Conclusion}
\label{sec:conc}

We have shown that the dichotomy criterion for conservative CSP can be decided in true polynomial time, without any assumption on the arity or the domain size of the input constraint language. This solves an important question on the complexity of c-CSP among the few that remain. On the way, we have also proved that classes of conservative constraint languages defined by linear strong Mal'tsev conditions admitting a semiuniform algorithm always have a tractable meta-problem. This result is a major step towards a complete classification of meta-problems in conservative languages and complements nicely the results of~\cite{chen2016asking}.

It is known that Proposition~\ref{prp:chen} does not hold in general if the linearity requirement on the Mal'tsev condition is dropped, as semilattices are NP-hard to detect even in conservative constraint languages despite having a uniform algorithm~\cite{Green2008a}. The same happens if the idempotency of the Mal'tsev condition is dropped instead~\cite{chen2016asking}. However, the mystery remains if the requirement for a uniform algorithm is loosened since no tractable idempotent strong linear Mal'tsev condition is known to have a hard meta-problem. This prompts us to ask if our result on conservative constraint languages can extend to the general case.

\begin{question}
Does there exist an idempotent strong linear Mal'tsev condition $\mathcal{M}$ that has a semiuniform polynomial-time algorithm but whose meta-problem is not in P, assuming some likely complexity theoretic conjecture?
\end{question}

A negative answer would imply a uniform algorithm for constraint languages with a Mal'tsev polymorphism, whose potential existence was discussed in~\cite{carbonnel:hal-01230681}. 

Finally we believe that our algorithm, by producing the coloured graph in polynomial time, would be very helpful in the design of a uniform algorithm that solves every tractable conservative constraint language (should one exist).

\begin{question}
Does there exist a uniform polynomial-time algorithm for the class of all tractable conservative constraint languages?
\end{question}

%The first step is to determine a conservative binary polymorphism $f$ such that for every 2-element set $B$, $f_{|B}$ is a semilattice if $B$ is red and $f_{|B}$ returns its first argument otherwise. The existence of $f$ is guaranteed by the Three Operations Theorem.
%
%Let $\Gamma$ be a $2$-conservative relational structure and let $\mathcal{P} = \mathcal{IP}^2(\Gamma) = (\mathcal{X},\mathcal{D},\mathcal{C})$. Once again, we assume that $1$-minimality has been enforced on $\mathcal{P}$. We can interpret $\mathcal{P}$ as the meta-problem associated to an unconstrained binary operation symbol and reuse the definitions and lemmas about closed sets of variables seen in the last section.
%
%Let $X_0, \ldots, X_\alpha$ denote the hierarchy of closed sets given by Lemma~\ref{lem:subsets}. In this particular case, $X_{i+1}$ contains exactly the variables of $X_i$ plus two variables $x_{f(a,b)},x_{f(b,a)}$ for some $B_{i+1} = \{a,b\} \subseteq \mathcal{D}$.
%
%\begin{lemma}
%If a solution to $X_i$ is known, then a solution to $X_{i+1}$ can be found in polynomial time.
%\end{lemma}

\bibliographystyle{plain}
\bibliography{MasterCSP}

\end{document}